\documentclass[runningheads]{llncs}

\usepackage{lipsum}
\usepackage{amsmath}
\usepackage{amsfonts}
\usepackage{amssymb}
\usepackage{graphicx}
\usepackage{enumitem}
\usepackage{color}
\usepackage[english]{babel}
\usepackage{setspace}
\usepackage[numbers]{natbib}

 \usepackage{url}

\DeclareMathOperator*{\argmax}{arg\,max}

\newtheorem{thm}{Theorem}
\newtheorem{lem}[thm]{Lemma}

\newtheorem{cor}[thm]{Corollary}

\newcommand{\x}{{\bf x}}
\newcommand{\y}{{\bf y}}
\newcommand{\z}{{\bf z}}

\newcommand{\SW}{\operatorname{SW}}
\newcommand{\Opt}{\operatorname{OPT}}

\begin{document}
%
%
%
\title{Pirates in Wonderland:\\[0.1cm] Liquid Democracy has Bicriteria Guarantees}
\titlerunning{Liquid Democracy has Bicriteria Guarantees}
%
\author{Jonathan A. Noel\inst{1} \and Mashbat Suzuki\inst{2} \and 
	Adrian Vetta\inst{2}
	}
\authorrunning{J. Noel et al.}
%
\institute{University of Victoria, Victoria, Canada \\
\email{noelj@uvic.ca}\\[0.2cm]	
\and
McGill University, Montreal, Canada \\
\email{mashbat.suzuki@mail.mcgill.ca, adrian.vetta@mcgill.ca}
}
\maketitle              
\begin{abstract}
Liquid democracy has a natural graphical representation, the delegation graph.
Consequently, the strategic aspects of liquid democracy can be studied as a game
over delegation graphs, called the liquid democracy game. 
Our main result is that this game has bicriteria approximation guarantees, in terms
of both rationality and social welfare. Specifically, we prove the price of stability for $\epsilon$-Nash equilibria
is exactly~$\epsilon$ in the liquid democracy game. 
\end{abstract}

\section{Introduction}
Liquid democracy is a form of direct and representative democracy, based on the concept of {\em delegation}.
Each voter has the choice of voting themselves or transferring (transitively) its vote to a trusted proxy.
Recent interest in liquid democracy, from both practical and theoretical perspectives, was sparked 
by the Pirate Party in Germany and its Liquid Feedback platform~\cite{BKN14}.
Similar initiatives have subsequently been undertaken by the 
Demoex Party in Sweden, the Internet Party in Spain, and the Net Party in Argentina.

There are many potential benefits of a transitive delegation mechanism. 
Participation may improve in quantity for several reasons. The system is easy to use and understand, induces low barriers 
to participation, and is inherently egalitarian: there is no distinction between voters and representatives; every one
is both a voter and a delegator. Participation may also improve in quality due to the flexibility to choose different 
forms of participation: voters can chose to be active participants on topics they are comfortable with or delegate on topics
they are less comfortable with. Accountability may improve due to the transparent nature of the mechanism and 
because there is a demonstrable line of responsibility between a delegated proxy and its delegators.
The quality of decision making may improve via a specialization to delegated experts
and a reduction in induced costs, such as the duplication of resources.

Our objective here is not to evaluate such claimed benefits, but we refer the reader to~\cite{Beh17, BKN14, BZ16, For02, GA15} for detailed 
discussions on the motivations underlying liquid democracy.
Rather, our focus is to quantitatively measure the performance of liquid democracy in an idealized setting.
Specifically, can equilibria in these voting mechanisms provide high social welfare?
That is, we study the {\em price of stability of liquid democracy}.

\subsection{Background}
As stated, vote delegation lies at the heart of liquid democracy. Furthermore, 
vote delegation in liquid democracy has several fundamental characteristics: optionality, retractability, 
partitionability, and transitivity. So let us begin by defining these concepts and tracing their origins~\cite{Beh17,BKN14}.

The notion of {\em optional} delegation proffers voters the choice of direct participation (voting themselves/choosing to abstain)
or indirect participation (delegating their vote). This idea dates back over a century to the work of Charles Dodgson on parliamentary 
representation~\cite{Dod1884}.\footnote{Dodgson was a parson and 
a mathematician but, as the author of  ``Alice in Wonderland'', is more familiarly known by his {\em nom de plume}, Lewis Carroll.}

Miller~\cite{Mil69} proposed that delegations be {\em retractable} and {\em partitionable}.
The former allows for delegation assignments to be time-sensitive and reversible.
The latter allows a voter to select different delegates for different policy decisions.\footnote{This option is particularly useful 
where potential delegates may have assorted competencies. For example, Alice may prefer to delegate to the Hatter on
matters concerning tea-blending but to the Queen of Hearts on matters concerning horticulture.}
 
Finally, {\em transitive} delegation is due to Ford~\cite{For02}. This allows a proxy to themselves delegate its vote {\em and} 
all its delegated votes. This concept is central to liquid democracy. Indeed, if an agent is better served by delegating her vote to 
a more informed proxy it would be perverse to prohibit that proxy from re-delegating that vote to an even more informed proxy.  
Moreover, such transitivity is necessary should circumstances arise causing the proxy to be unable to vote. It also reduces the 
duplication of efforts involved in voting.

As noted in the sixties by Tullock~\cite{Tol67}, the development of the computer opened up the possibility of large proxy voting systems.
Indeed, with the internet and modern security technologies, liquid democracy is inherently practical; see Lumphier~\cite{Lan95}.

There has been a flurry of interest in liquid democracy from the AI community. This is illustrated by the large range of 
recent papers on the topic; see, for example, \cite{Bri18,  BT18, CG17, CMM17, GKM18, 
GA15, KMP18, KR20, ZZ17}. Most directly related to our work is the game theoretic model of liquid democracy 
studied by Escoffier et al.~\cite{EGP19}. (A related game-theoretic model was also investigated by Bloembergen et al.~\cite{BGL19}.) Indeed, our motivation is an open question posed by Escoffier et al.~\cite{EGP19}:
are price of anarchy type results obtainable for their model of liquid democracy?
We will answer this question for a generalization of their model. 


\subsection{Contributions}

In Section~\ref{sec:model}, we will see that vote delegation has a natural representation in terms of a 
directed graph called the {\em delegation graph}. If each agent $i$ has a utility of $u_{ij}\in [0,1]$ when agent $j$ 
votes as her delegate then a game, called the {\em liquid democracy game}, is induced on the delegation graph.
We study the {\em welfare ratio} in the liquid democracy game, which compares the social welfare of an equilibrium 
to the welfare of the optimal solution.

Pure strategy Nash equilibria need not exist in the liquid democracy game, so we focus on mixed 
strategy Nash equilibria.
Our main result, given in Section~\ref{sec:bicriteria} is that bicriteria approximation guarantees (for social welfare and rationality) 
exist in the game.
\begin{thm}\label{thm:upper}
For all $\epsilon\in [0,1]$, and for any instance of the liquid democracy game, there exists an $\epsilon$-Nash equilibrium with social 
welfare at least $\epsilon\cdot \Opt$.
\end{thm}
Theorem~\ref{thm:upper} is tight: the stated bicriteria guarantees cannot be improved.
\begin{thm}\label{thm:lower}
For all $\epsilon\in[0,1]$, there exist instances such that any $\epsilon$-Nash equilibrium has welfare at most $\epsilon\cdot \Opt +\gamma$ for any $\gamma>0$.
\end{thm}

Theorems~\ref{thm:upper} and~\ref{thm:lower} imply
that the {\em price of stability} for $\epsilon$-Nash equilibria is~$\epsilon$.
An important consequence of Theorem~\ref{thm:upper} is that strong approximation guarantees can {\em simultaneously}
be obtained for both social welfare and rationality.
Specifically, setting~$\epsilon=\frac12$ gives factor $2$ approximation guarantees for each criteria.
\begin{cor} \label{cor:upper} 
For any instance of the liquid democracy game, there exists a $\frac12$-Nash equilibrium with welfare at least~$\frac12\cdot \Opt$.
\end{cor}

\section{A Model of Liquid Democracy}\label{sec:model}

In this section, we present the liquid democracy game. This game generalizes the game-theoretic 
model studied by Escoffier et al.~\cite{EGP19}. 

\subsection{The Delegation Graph}\label{sec:delegation-graph}
In liquid democracy each agent has three strategies: she can abstain, vote herself, or delegate her 
vote to another agent. So we can represent an instance by a directed network $G=(V,A)$
called the {\em delegation graph}.
There is a vertex in $V=\{1,\cdots, n\}$ for each agent. 
To define the sets of arcs, there are three possibilities.
First, if agent $i$ votes herself the delegation graph contains a self-loop $(i,i)$.
Second, if agent $i$ delegates her vote to agent $j\neq i$ 
then there is an arc $(i,j)$ in $G$.
Third, if agent $i$ abstains then the vertex $i$ has out-degree zero.

Now, because the out-degree of each vertex is at most one, the delegation graph $G$ is a $1$-forest.  
That is, each component of $G$ is an arborescence plus at most one arc.
In particular, each component is either an arborescence and, thus, contains no cycle or 
contains exactly one directed cycle (called a {\em delegation cycle}).
In the former case, the component contains one sink node corresponding to an abstaining voter.
In the latter case, if the delegation cycle is a self-loop the component contains exactly one voter called
a {\em guru}; if the delegation cycle has length at least two then the component contains no voters. 
An example of a delegation graph is shown in Figure~\ref{fig:delegation-graph}.
\begin{figure}
	\begin{center}
	\includegraphics*[scale=0.13]{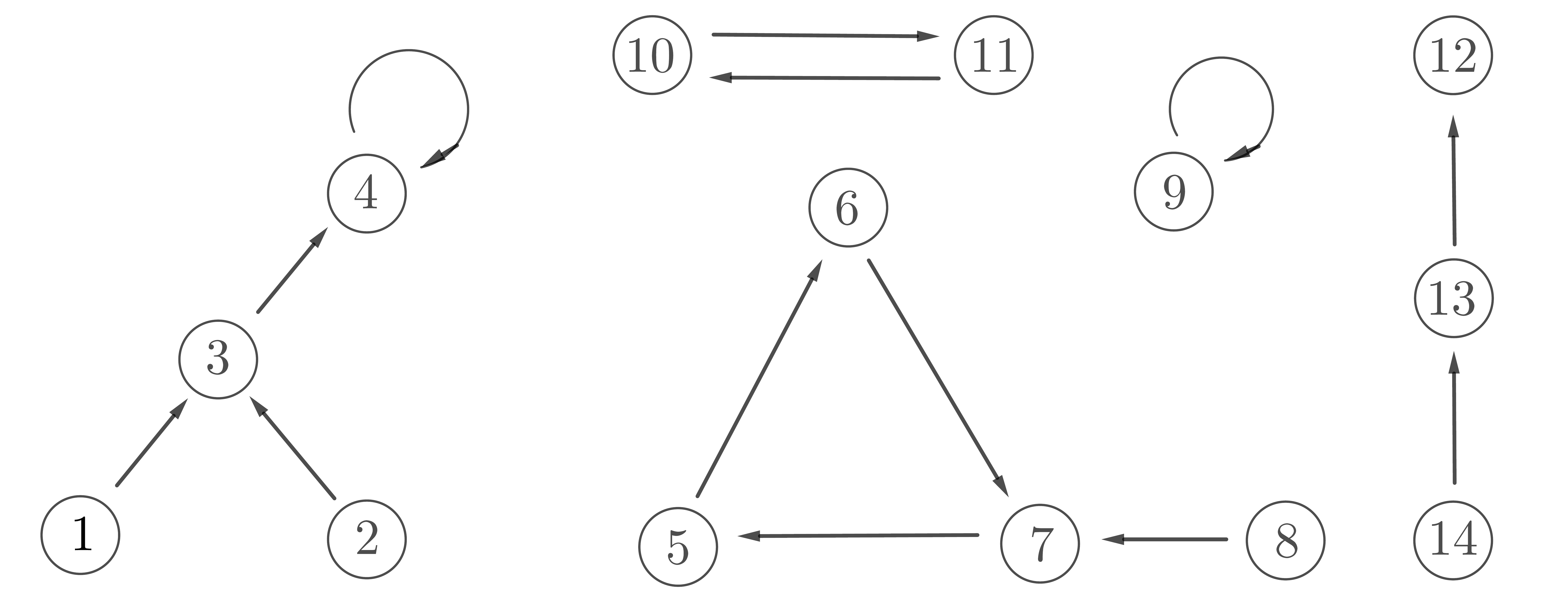}
	\end{center}	
	\caption{A delegation graph.}\label{fig:delegation-graph}	
\end{figure}

Observe that, by the transitivity of delegations, if an agent $i$ is in a component containing
a guru $g$ then that guru will cast a vote on $i$'s behalf. On the other hand, if agent $i$ is 
in a component without a guru (that is, with either a sink node or a cycle of length at least two)
then no vote will be cast on $i$'s behalf. 
We denote the guru $j$ representing agent $i$ by $g(i)=j$ if it exists (we write $g(i)=\emptyset$ otherwise).
Furthermore, its easy to find $g(i)$: simply apply path traversal starting at vertex $i$ in
the delegation graph. 

For example, in Figure~\ref{fig:delegation-graph} two components contain a guru. 
Agent $4$ is the guru of agents $1,2,3$ and itself; agent $9$ is the guru only for itself. 
The vertices in the remaining three components have no gurus.
There are two components with {\em delegation cycles}, namely $\{(10,11), (11, 10)\}$ and 
$\{(5,6), (6, 7), (7,5)\}$. The final component also contains no guru as agent $12$ is a sink node and thus abstains.

\subsection{The Liquid Democracy Game}\label{sec:delegation-game}
The game theoretic model we study is a generalization of the model of Escoffier et al.~\cite{EGP19}.
A pure strategy $s_i$ for agent $i$ corresponds to the selection of at most one outgoing arc.
Thus we can view $s_i$ as an $n$-dimensional vector $\x_i$. Specifically, 
$\x_i$ is a single-entry vector with entry $x_{ij}=1$ if $i$ delegates to agent $j$ and $x_{ij}=0$ otherwise.
Note that if $x_{ii}=1$ then $i$ votes herself (``delegates herself'') and that $\x_i={\bf 0}$ if agent $i$ abstains.

It immediately follows that there is a unique delegation graph $G_{\x}$ associated with
a pure strategy profile $\x=(\x_1, \x_2 \cdots, \x_n)$.
To complete the description of the game, we must define the payoffs corresponding to
each pure strategy profile $\x$. To do this, let agent $i$ have a utility $u_{ij}\in [0,1]$ if she 
has agent $j$ as her guru. Because of the costs of voting in terms of time commitment, knowledge acquisition, etc.,
it may be that $u_{ii}< u_{ij}$ for some $j\neq i$ trusted by $i$.\footnote{Indeed, if this is not the case then
liquid democracy has no relevance.}
We denote the utility of agent $i$ in the delegation graph $G_{\x}$ by $u_i(\x)=u_{i,g(i)}$. 
If agent~$i$ has no guru then it receives zero utility.\footnote{We remark that our results hold even when agents who abstain or have no guru obtain positive utility.}
It follows that only agents that lie in a component of $G_{\x}$ containing a guru can obtain positive utility.

For example, in Figure~\ref{fig:delegation-graph} agent $9$ is a guru so receives utility $u_{9,9}$.
Each agent $i\in \{1,2,3,4\}$ has agent $4$ as its guru so receives utility $u_{i,4}$. 
All the remaining agents have no guru and so receive zero utility.   

An agent $i$ is playing a best response at a pure strategy $\x=(\x_1, \x_2, \cdots, \x_n)$ if he cannot increase
his utility by selecting a different or no out-going arc. The strategy profile is a pure Nash equilibrium if every agent is
playing a best response at $\x$. 
Our interest is in comparing the social welfare of equilibria to the optimal welfare in the game.
To do this, let the social welfare of $\x$ be $\SW(\x) = \sum_{i\in V} u_i(\x)$ and let 
$\Opt=\max\limits_{\x} \sum_{i\in V} u_i(\x)$ be the optimal welfare over all strategy profiles. 
The {\em price of stability} is the worst ratio over all instances between the {\em best} welfare of a Nash equilibrium
and the optimal social welfare. 

The reader may ask why could an equilibrium have low social welfare. The problem is that delegation is transitive, but {\em trust is not}.
 Agent $i$ may delegate to an agent $j$ where $u_{ij}$ is large but $j$ may then re-delegate to an agent $k$ where $u_{ik}$ is small.
Worse, agent $i$ may receive no utility if the transitive delegation of its vote leads to a delegation cycle or an abstaining voter.
 Unfortunately, not only can pure Nash equilibria have low social welfare in the liquid democracy game they need not even exist!
 \begin{lem}\label{lem:no-PSNE}
There exist liquid democracy games with no pure strategy Nash equilibrium. 
\end{lem}
\begin{proof}
Let there be three voters with $u_1=(\frac{1}{2},1,0)$, $u_2=(0,\frac{1}{2},1)$ and $u_3=(1,0,\frac{1}{2})$. 
(A similar instance was studied in \cite{EGP19}.) Assume $\x$ is a pure Nash equilibrium and let $S$ be the set of gurus 
in $G_{\x}$. There are two cases.
First, if $|S|\leq 1$ then there exists an agent $i$ with zero utility. This agent can deviate and vote herself to 
obtain utility of $\frac{1}{2}>0$, a contradiction. 
Second, if $|S|\geq 2$ then one of the gurus can delegate its vote to another guru to obtain a utility of $1> \frac{1}{2}$, a contradiction.
Therefore, no pure Nash equilibrium exists. \qed 
\end{proof}

\subsection{Mixed Strategy Equilibria}\label{sec:MSNE}

Lemma~\ref{lem:no-PSNE} tells us that to obtain performance guarantees we must look beyond pure strategies.
Of course, as liquid democracy is a finite game, a mixed strategy Nash equilibrium always exists.
It follows that we can always find a mixed Nash equilibrium in the liquid democracy game and compare
its welfare to the optimal welfare.

Here a mixed strategy for agent $i$ is now simply a non-negative vector $\x_i$ where the entries 
satisfy $\sum_{j=1}^n x_{ij}\le 1$. Note that agent $i$ then abstains with probability $1-\sum_{j=1}^n x_{ij}$.
However, because abstaining is a weakly dominated strategy, we may assume no voter abstains in either the
optimal solution or the best Nash equilibria. Hence, subsequently, each $\x_i$ will be unit-sum vector.

What is the utility $u_i(\x)$ of an agent for a mixed strategy profile $\x=(\x_1, \x_2, \dots, \x_n)$?
It is simply the expected utility  given the probability distribution over the delegation graphs generated by $\x$.
There are several equivalent ways to define this expected utility. For the purpose of analysis, 
the most useful formulation is in terms of directed paths in the delegation graphs.
Denote by $\mathcal{P}(i,j)$ be the set of all directed paths from $i$ to $j$ in a complete directed graph 
on $V$. Let $\mathbb{P}_{\x}( g(i)=j)$ be the probability that $j$ is the guru of $i$ given the mixed strategy profile $\x$. Then
\begin{align*}
	u_i(\x) &= \sum_{j\in V} \mathbb{P}_{\x}( g(i)=j  ) \cdot u_{ij} \\ 
	 &= \sum_{j\in V} \mathbb{P}_{\x}( \exists \text{ path from } i \ \text{to} \ j  )\cdot \mathbb{P}_{\x}(\exists \text{ arc } (j,j) ) \cdot u_{ij}\\
	&=\sum_{j\in V} \left( \sum_{P \in \mathcal{P}(i,j)}  \prod\limits_{a\in P}  x_a \right)  \cdot x_{jj} \cdot u_{ij}
\end{align*}
To understand this recall that $\x_i$ is a unit-sum vector. It follows that each delegation graph generated by
the mixed strategy $\x$ has uniform out-degree equal to one.
This is because we generate a delegation graph from $\x$ by
selecting exactly one arc emanating from $i$ according to the probability distribution $\x_i$.
Consequently, $j$ can be the guru of $i$ if and only if the delegation graph contains a self-loop $(j,j)$ {\em and} contains 
a unique path $P$ from $i$ to $j$. The second equality above then holds.
Further, the choice of outgoing arc is independent at each vertex $i$. This implies the third equality.

Regrettably, however, no welfare guarantee is obtainable even with mixed strategy Nash equilibria.

 \begin{lem}\label{lem:bad-MSNE}
The price of stability in liquid democracy games is zero. 
\end{lem}
\begin{proof}
	
	Consider an instance with three agents whose utility is given as $u_1=(\delta,1,0),u_2=(0,\delta,1) $ and $u_3=(1,0,\delta)$. The optimal welfare is $\Opt=1+2\delta$ obtained by the first and second agent voting and the third agent delegates to the first agent.  
	
	By a similar argument to Lemma~\ref{lem:no-PSNE}, we know that no pure strategy Nash equilibrium exists. It is straightforward to verify that this game has a unique mixed strategy Nash equilibria $\x$, where $x_1=(\delta,1-\delta,0)$, $x_2=(0,\delta,1-\delta)$ and $x_3 =(1-\delta, 0, \delta) $. 
	
 The expected utility of agent~$1$ is then
\begin{eqnarray*}
u_1(\x) &=& \sum_{j\in V} \left( \sum_{P \in \mathcal{P}(i,j)}  \prod\limits_{a\in P}  x_a \right)  \cdot x_{jj} \cdot u_{ij} \\
&=& x_{11} \cdot u_{11} + x_{12}\cdot x_{22} \cdot u_{12}  \\
&=& \delta^2 + (1-\delta)\cdot \delta \\
&=& \delta
 \end{eqnarray*}
The case of agents $2$ and $3$ are symmetric, hence the social welfare of this equilibria is $3\delta$. 
Thus, the price of stability is $\frac{3\delta}{1+2\delta}$ which tends to zero
as $\delta\rightarrow 0$. \qed 
\end{proof}

Lemma~\ref{lem:bad-MSNE} appears to imply that no reasonable social welfare guarantees
can be obtained for liquid democracy. This is not the case. Strong performance guarantees
can be achieved, provided we relax the incentive constraints.
Specifically, we switch our attention to approximate Nash equilibria.
A strategy profile $\x$ is an $\epsilon$-{\em Nash equilibrium} if, for each agent $i$, 
$$u_{i}(\x)\geq \ (1-\epsilon)\cdot u_{i}(\hat{\x}_i,\x_{-i}) \qquad \forall \hat{\x}_i $$

Above, we use the notation $\x_{-i}=\{\x_i\}_{j\neq i}$.
Can we obtain good welfare guarantees for approximate Nash equilibria? We will prove the answer is {\sc yes}
in the remainder of the paper. In particular, we present tight bounds on the price of stability for $\epsilon$-Nash equilibria.

\section{The Price of Stability of Approximate Nash Equilibria}\label{sec:bicriteria}

So our task is to compare the social welfare of the best $\epsilon$-Nash equilibrium with
the social welfare of the optimal solution. Let's begin by investigating the optimal solution.

\subsection{An Optimal Delegation Graph} 
By the linearity of expectation, there is an optimal solution in which the agents use only pure strategies. 
In particular, there exists an {\em optimal} delegation graph maximizing social welfare. Moreover, this optimal graph 
has interesting structural properties. 
To explain this, we say that agent $i$ is {\em happy} if she has strictly positive utility in a delegation graph $G$, that is 
$u_{i,g(i)}>0$.  A component $Q$ in $G$ is {\em jolly} if every vertex in $Q$ (except, possibly, the guru) is happy.

The key observation then is that there is an optimal delegation graph in which every component is a jolly star.

\begin{lem}\label{lem: OptLD}
There is an optimal delegation graph that is the disjoint union of jolly stars.
\end{lem}
\begin{proof}
Let $Q$ be a component in an optimal delegation graph $G$. We may assume $Q$ contains a guru.
To see this, suppose $Q$ contains an abstaining sink node $j$. Then the graph $\hat{G}=G\cup (j,j)$ where $j$ votes 
herself is also optimal. On the other hand, suppose $Q$ contains a cycle $C$ of length at least $2$. Take an agent $j\in C$. 
Then the graph $\hat{G}$ obtained by $j$ voting herself instead of delegating her vote is also optimal.

So we may assume each component contains a guru. Further, we may assume each component is a star.
Suppose not, take a component $Q$ with guru $j$ containing an agent $i$ that does not delegate to $j$.
But then the graph $\hat{G}$ obtained by $i$ delegating her vote directly to $j$ is optimal.

Finally, we may assume each star is {\em jolly}. Suppose $i$ is not a happy agent in a star $Q$ with guru $j$.
Thus $u_{ij}=0$. But then if $i$ changes her delegation and votes herself we again obtain an optimal solution.
In this case $i$ will form a new singleton component which is trivially a jolly star. \qed
\end{proof}

Lemma~\ref{lem: OptLD} states that the optimal solution can be obtained by a pure strategy~$\x^{*}$ whose 
delegation graph is union of jolly stars. The centre of each star is a guru in the optimal solution and the leaves are the 
corresponding happy agents who delegated to the guru.  Denote the set of gurus in the optimal solution by $D^{*}=\{i\in V \ : \ x^*_{ii}=1\}$,  
and let $L_j = \{i\in V\setminus D^{*} \ : \ x^*_{ij}=1\}$ be the agents who delegate to the guru $j$ as illustrated in Figure~\ref{Fig: Star}. It follows that the optimal solution has welfare
$$
\Opt = \sum_{j\in D^{*} } \left( u_{jj} + \sum_{i\in L_j}u_{ij}\right)  
$$ 

\begin{figure}[bt!]
	\centering
	\includegraphics[scale=0.14]{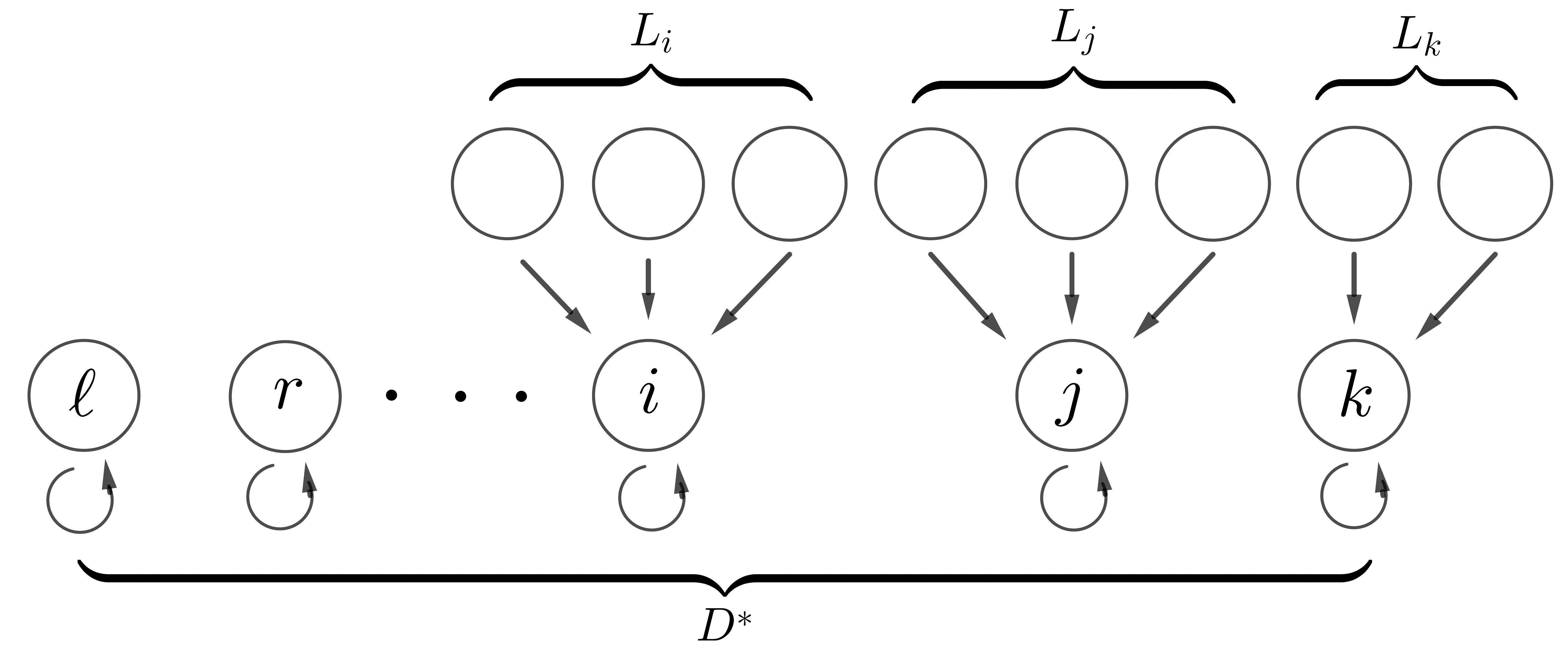}
	\caption{Welfare optimal delegation graph that is disjoint union of jolly stars. }
	\label{Fig: Star}
\end{figure}
\subsection{A Stable Solution}

In order to study the best $\epsilon$-Nash equilibrium we now show the existence of
a ``potentially'' stable solution $\x$ whose definition is inspired by the set $D^*$ of gurus in the optimal solution. 
In Section~\ref{sec:bounds} we will prove that $\x$ is indeed an $\epsilon$-Nash equilibrium and also has high
social welfare.

To obtain $\x$ we require the following definitions.
Denote the standard set of feasible mixed strategies for agent $i$ as 
$$\mathbb{S}_i=\{ \x_{i}\in \mathbb{R}^n_+ \ : \ \sum\limits_{j\in V}x_{ij}=1\  \}$$
Given a fixed strategy profile $\x_{-i}=\{\x_j\}_{j\neq i}$ for the other agents, let
the corresponding best response for agent $i$ be 
$$B_i(\x_{-i})=\argmax\limits_{\hat{\x}\in \mathbb{S}_i} u_{i}(\hat{\x},\x_{-i})$$
For each $i\in D^*$, we denote a restricted set of mixed strategies 
$$\mathbb{S}^R_i = \{ \x_{i}\in \mathbb{R}^n_+ \ : \ \sum\limits_{j\in V}x_{ij}=1\ , x_{ii}\geq \epsilon \}$$
Then for a fixed strategy profile $\x_{-i}$, let 
$$B^{R}_i(\x_{-i})=\argmax\limits_{\hat{\x}\in \mathbb{S}^R_i } u_{i}(\hat{\x},\x_{-i})$$ 
be the best response for the agent $i$ from amongst the restricted set of feasible strategies. 
Next recall Kakutani's Fixed Point Theorem.
\begin{thm}[Kakutani's Fixed Point Theorem] 
Let $K$ be a non-empty, compact and convex subset of $\mathbb{R}^m$, and let $\Phi:K \rightarrow 2^K$ be a set-valued 
function on~$K$ such that:\\
\indent (i) $\Phi(x)$ is non-empty and convex for all $x\in K$, and\\
\indent (ii) $\Phi$ has a closed graph.\\
Then $\Phi$ has a fixed point, that is, there exists an $x^*\in K$ with $x^*\in \Phi(x^*)$.
\end{thm}	
Here a set-valued function $\Phi$ has a {\em closed graph} if $(x^k,y^k)\rightarrow (x,y) $ and $y^k \in \Phi (x^k) $ implies that $y\in \Phi(x)$.

\begin{thm}\label{thm: BR}
There exists a strategy profile $\x$ such that:\\
\indent (a) For all $i\in D^*$, we have $\x_i \in B^R_i(\x_{-i})$, and\\
\indent (b) For all $j\notin D^*$, we have $\x_j \in B_j(\x_{-j})$.
\end{thm}	
\begin{proof}
	
Let the feasible set of strategy profiles be $\Xi= \prod\limits_{i\in D^*} \mathbb{S}^R_i \times \prod\limits_{j\notin D^*} \mathbb{S}_j$, 
a subset of Euclidean space. 
Without loss of generality, let $D^*=\{1,2,\dots, k\}$.
Now define a set valued function $\Phi \colon \Xi \longrightarrow 2^\Xi$ by
$$\x \mapsto ( \underbrace{B_1^R(\x_{-1}),\cdots,B_k^R(\x_{-k}),}_{D^*} \underbrace{B_{k+1}(\x_{-(k+1)}),\cdots,B_n(\x_{-n})}_{V\setminus D^*}    )$$	
That is, for each $\x\in \Xi$ we have $\Phi(\x)\subseteq \Xi$. Note the statement of the theorem is equivalent to showing that $\Phi$ has a fixed point. 
	
Observe that $\Phi$ satisfies the conditions of Kakutani's Fixed Point Theorem. Indeed $\Xi$ is nonempty, 
compact and convex, since it is a product of non-empty, compact and convex sets $\mathbb{S}^R_i$ and $\mathbb{S}_j$. 
	
Next let's verify that $\Phi(\x)\neq \emptyset$. This holds since, for each agent $i$,  we have $B^R_i(\x_{-i})\neq \emptyset$ 
or $B_i(\x_{-i})\neq \emptyset$ by the continuity of $u_i(\ \cdot \ ,\x_{-i})$ and the Weierstrass Extreme Value Theorem. 
	
Furthermore, for all $\x\in \Xi$ the set $\Phi(\x)\subseteq \Xi$ is convex. This is because, for each $i\in D^*$ 
and $j\in V\setminus D^*$, the sets $B^R_i(\x_{-i})$ and $B_j(\x_{-j})$ are convex, and thus $\Phi(\x)$ is Cartesian 
product of convex sets. We must now show that both $B_j(\x_{-j})$ and $B^R_i(\x_{-i})$ are convex.
The convexity of $B_j(\x_{-j})$ follows immediately by the multilinearity of $u_i$. 
Next take an agent $i\in D^* $. If $\y_i,\z_i \in B^R_i$ then, 
for all $\lambda\in [0,1]$ and any $\hat{\x_i}\in \mathbb{S}^R_i$, we have
	\begin{align*}
	u_i(\lambda\y_i+(1-\lambda)\z_i,\x_{-i}) &\ =\ \lambda u_i(\y_i,\x_{-i}) + (1-\lambda)u_i(\z_i,\x_{-i}) \\ 
	&\ \geq \ u_i (\hat{\x}_i,\x_{-i})
	\end{align*}
Observe $\lambda\y_i+(1-\lambda)\z_i\in \mathbb{S}^R_i$ since $\lambda y_{ii}+(1-\lambda)z_{ii}\geq \lambda \epsilon +(1-\lambda)\epsilon = \epsilon$ 
and $\lambda\sum_{j\in V} y_{ij}+(1-\lambda)\sum_{j\in V} z_{ij} = 1  $. Thus $\lambda \y_i +(1-\lambda)\z_i\in B^R_i(\x_{-i})$ for 
any $\lambda\in [0,1]$, which implies $B^R_i(\x_{-i})$ is convex.
	
Finally, $\Phi$ has a closed graph because each $u_i(\x_i,\x_{-i})$ is a continuous function of $\x_i$ for any fixed $\x_{-i}$, and 
both sets $\mathbb{S}^R_i$ and $\mathbb{S}_i$ are compact. Thus, by Kakutani's Fixed Point Theorem, $\Phi$ 
has a fixed point $\x$. Hence (a) and (b) hold. \qed	
\end{proof}

\subsection{Bicriteria Guarantees for Liquid Democracy}\label{sec:bounds}

We will now prove that the fixed point $\x$ from Theorem~\ref{thm: BR} gives our main result: there is an $\epsilon$-Nash equilibrium
with welfare ratio at least $\epsilon$. First let's show the incentive guarantees hold.
\begin{lem}\label{lem:approx-NE}
The fixed point $\x$ is an $\epsilon$-Nash equilibrium.
\end{lem} 
\begin{proof}
Take an agent $j\notin D^*$. Then for any $\hat{\x}_j \in \mathbb{S}_i$ we have
$$u_{j}(\x_j,\x_{-j})\ \geq\ u_{j}(\hat{\x}_j,\x_{-j}) \ \geq\ (1-\epsilon)\cdot u_i(\hat{\x}_i,\x_{-i})$$ 
The incentive guarantee for $j$ follows immediately.
%

Next consider an agent $i\in D^*$. Take any $\hat{\x}_i\in \mathbb{S}_i$ and define a new strategy $\y_i$ as follows:
		$$
		y_{ij}=
		\begin{cases}
		\epsilon+ (1-\epsilon) \cdot \hat{x}_{ii} &\text{if} \ \ j=i \\ 
		(1-\epsilon)\cdot  \hat{x}_{ij}  &\text{if} \ \ j\neq i \\ 
		\end{cases}
		$$ 
Observe that $\y_i\in \mathbb{S}^R_i$ because $y_{ii}\geq \epsilon$ and 
$\sum_{j\in N} y_{ij}= \epsilon+(1-\epsilon)\cdot \sum_{j\in V} \hat{x}_{ij} =1$. 
		
Now for any path $P=\{a_1,a_2,\cdots, a_k\}$ from $i$ to $j$ where $a_i$ are the arcs, the probability of obtaining 
this path in the delegation graph generated by the strategy profile $\{\y_i, \x_{-i}\}$ is exactly $y_{a_1}\cdot  \prod\limits_{a\in P\setminus a_1} x_a$. 
Thus
\begin{align*}
		u_i(\y_i,\x_{-i}) 
		&\ =\  \sum_{j\in V} \mathbb{P}_{\y_i,\x_{-i}}( g(i)=j  ) \cdot u_{ij} \\
		&\ =\ u_{ii}y_{ii} + \sum_{j\in V\setminus i} u_{ij} x_{jj}  
			\cdot \left( \sum_{P \in \mathcal{P}(i,j)} y_{a_1}  \prod\limits_{a\in P\setminus a_1} x_a\right) \\  	   		
		&\ \geq\ (1-\epsilon)\cdot u_{ii}\hat{x}_{ii} + (1-\epsilon)\cdot \sum_{j\in V\setminus i} u_{ij} x_{jj} 
		\cdot \left( \sum_{P \in \mathcal{P}(i,j)} \hat{x}_{a_1}  \prod\limits_{a\in P\setminus a_1} x_a\right) \\ 
		&\ =\ (1-\epsilon)\cdot \left(u_{ii}\hat{x}_{ii} + \sum_{j\in V\setminus i} u_{ij} x_{jj}  
		\cdot \left( \sum_{P \in \mathcal{P}(i,j)} \hat{x}_{a_1}  \prod\limits_{a\in P\setminus a_1} x_a\right)\right) \\
		&\ =\ (1-\epsilon) \cdot u_{i}(\hat{\x}_i,\x_{-i})
\end{align*}
But $\x_i\in B^R_i(\x_{-i})$. Hence, $u_i(\x_i,\x_{-i})\geq u_i(\y_i,\x_{-i})$ because $\y_i\in \mathbb{S}^R_i$.
It follows that $u_i(\x_i,\x_{-i})\geq (1-\epsilon)\cdot u_i(\hat{\x}_i,\x_{-i})$
and so the incentive guarantee for $i$.
Thus $\x$ is an $\epsilon$-Nash equilibrium. \qed  	
\end{proof}

Next let's prove the social welfare guarantee holds for $\x$.
\begin{thm}\label{thm:bicriteria}
For all $\epsilon\in [0,1]$, and for any instance of the liquid democracy game, there exists an $\epsilon$-Nash equilibrium with social 
welfare at least $\epsilon\cdot \Opt$.
\end{thm}
\begin{proof}
It suffices to prove that $\x$ has social welfare at least $\epsilon\cdot \Opt$.
We have	   
	\begin{align*}
	\SW(\x) &= \sum_{j\in V} u_{j}(\x)  \\ 
	& = \sum_{j\in D^*} u_j(\x) + \sum_{j\in D^*} \sum_{i\in L_j} u_i(\x) \\ 
	&\geq \sum_{j\in D^*}u_{jj}\cdot  x_{jj} +  \sum_{j\in D^*} \sum_{i\in L_j} u_{ij}\cdot x_{jj} \\ 
	&\geq \sum_{j\in D^*} \left( u_{jj}\cdot \epsilon  +\sum_{i\in L_j} u_{ij}\cdot \epsilon \right)  \\ 
	&= \epsilon\cdot \sum_{j\in D^*} \left( u_{jj} +  \sum_{i\in L_j} u_{ij} \right)  \\ 
	&= \epsilon\cdot \Opt
	\end{align*}	
The first inequality follows since each agent $i\in L_j$ satisfies $u_i(\x_i,\x_{-i})\geq u_i(\hat{\x}_i,\x_{-i})$ for all 
$\hat{\x}_i \in \mathbb{S}_i$. 
In particular, the deviation $\hat{\y}_i$ of delegating to the guru $j$ with probability $1$ 
implies $u_i(\x_i,\x_{-i})\geq u_i(\hat{\y}_i,\x_{-i})\geq  u_{ij}x_{jj}$. Finally, the second inequality holds as
we have $\x_j\in  \mathbb{S}_j^R$, for each $j\in D^*$. Therefore $x_{jj}\geq \epsilon$ and the
welfare guarantee holds.~\qed
\end{proof}
We can deduce from Theorem~\ref{thm:bicriteria} that strong approximation guarantees can {\em simultaneously}
be obtained for both social welfare and rationality. In particular, setting~$\epsilon=\frac12$ gives factor $2$ approximation guarantees for both criteria.

\begin{cor} \label{cor:upper} 
For any instance of the liquid democracy game, there exists a $\frac12$-Nash equilibrium with welfare at least~$\frac12\cdot \Opt$.
\end{cor}

\subsection{A Tight Example}\label{sec:tight}
We now prove upper bounds on the welfare guarantee obtainable by \textit{any} $\epsilon$-Nash equilibrium. 
In particular, we show that the bicriteria guarantee obtained in Theorem~\ref{thm:bicriteria} is tight.  
\begin{thm}\label{thm:PoS-upper}
For all $\epsilon\in[0,1]$, there exist instances such that any $\epsilon$-Nash equilibrium has 
welfare at most $\epsilon\cdot \Opt +\gamma$ for any $\gamma>0$
\end{thm}
\begin{proof}
	Take an instance with $n+2$ agents. Let agents $\{1,2\}$ have identical utility functions with $u_{ij}=\delta $ if $j=1$ and $u_{ij}=0$ otherwise. 
	The remaining agents $i\in\{3,\cdots, n+2\}$ have utilities  	
	$$
	u_{ij} = 
	\begin{cases}
	1 &\text{if } j=2  \\ 
	0 & \text{otherwise}	
	\end{cases} 
	$$ 
	Observe that $\Opt = \delta + n$ which is obtained by agents 1 and 2 voting while the remaining 
	agents delegate to agent $2$. Now let $\x$ be an $\epsilon$-Nash equilibrium. We claim that $x_{22}\leq \epsilon$.
	To see this, note that $x_{11}\geq (1-\epsilon)$.  If $x_{11}< (1-\epsilon)$ then $u_{1}(\x)<(1-\epsilon)\delta $. 
	But this contradicts the fact that $\x$ is an $\epsilon$-Nash equilibrium, as agent $1$ can deviate and vote herself to obtain a utility of $\delta$.
	Furthermore,  $x_{11}\geq (1-\epsilon)$ implies  $x_{21}\geq (1-\epsilon)$ by a similar argument. 
	Since $\sum_{j\in N} x_{2j}=1$, we do have $x_{22}\leq \epsilon$. The social welfare of $\x$ is then
	$$
	\SW(\x) \ \leq\ \delta + (1-\epsilon)\delta + \epsilon n  
	\ =\ 2(1-\epsilon)\delta + \epsilon \Opt
	$$
	Letting $\gamma= 2(1-\epsilon)\delta$ gives the desired bound. Since $\delta$ can be made arbitrarily small, $\gamma$ 
	can also be made arbitrarily small.  \qed 
\end{proof}
Together, Theorems~\ref{thm:bicriteria} and~\ref{thm:PoS-upper} imply that the price of stability of $\epsilon$-Nash equilibrium is exactly $\epsilon$ in the liquid democracy game.

\section{ Extensions and Computational Complexity}

\subsection{Model Extensions}

Our bicriteria results extend to the settings of repeated games, weighted voters, and multiple delegates.

\noindent $\bullet$ {\em Repeated Games}.
Recollect that an underlying motivation for liquid democracy is that it can be applied repeatedly over time with agents
having the option of delegating to different agents at different times. Evidently, by repeating the delegation
game over time, with different utility functions for the topics considered in different time periods, 
the same bicriteria guarantees holds.

\noindent $\bullet$ {\em Weighted Voters}. In some electoral systems, different agents may have different voting powers.
That is, the voters are weighted. In this case, the number of votes a guru casts is simply the sum of the weights of all the
votes to which it was delegated. Our bicriteria guarantees then follow trivially.

\noindent $\bullet$ {\em Multiple Delegates}.
In some settings it may be the case that an agent is allowed to nominate more than one delegate.
That are two natural models for this. First, an agent nominates multiple delegates but the mechanism
can use only one of them. Second, an agent splits its voting weight up and assigns it to multiple delegates.
In both cases our techniques can be adapted and applied to give the same performance guarantees.

\subsection{Computational Aspects}

Let us conclude by discussing computation aspects in the liquid democracy game.
Recall, to find an $\epsilon$-Nash equilibrium with a provably optimal social welfare guarantee
we solved a fixed point theorem. 
Moreover, solving the fixed point theorem requires knowledge of the optimal set $D^*$ of gurus.
But obtaining $D^*$ is equivalent to finding an optimal solution to the liquid democracy game 
and this problem is hard.
\begin{thm}\label{thm:hardness}
It is NP-hard to find an optimal solution to the liquid democracy game.
\end{thm}
\begin{proof}
We apply a reduction from {\em dominating set in a directed graph}. Given a directed 
graph $G=(V,A)$ and an integer $k$: is there a set $S\subseteq V$ of cardinality $k$ such that,
for each $i\notin S$, there exists an arc $(i,j)\in A$ for some $j\in S$. 
This problem is NP-complete. We now give a reduction to the liquid democracy game.

Given the directed graph $G=(V,A)$ let the set of agents in the game be $V$. Define the
utility function of an agent by:
$$
u_{ij} = 
\begin{cases}
1 &\text{if } (i,j)\in A  \\ 
0 & \text{otherwise}	
\end{cases} 
$$ 
It immediately follows that there is a dominating set of cardinality at most $k$ if and only if the
liquid democracy game has a solution with social welfare at least $n-k$.
This completes the proof.
\qed
\end{proof}

So is it possible for the agents to compute a good bicriteria solution in polynomial time? 
If we allow for sub-optimal approximation guarantees then this is achievable.
The idea is simple: the characteristics required of the agents to ensure 
reasonable performance guarantees are {\em narcissism} and {\em avarice}.\footnote{This is not an unreasonable assumption for
both pirates and many of the inhabitants of Wonderland!}
First, since $D^*$ is unknown, each agent $i$ narcissistically assumes he himself is an optimal guru.
Consequently, he will vote with probability $p$. Thus, he will delegate his vote 
with probability $(1-p)$. This he will do avariciously, by greedily delegating to the agent $i^*$
that gives him largest myopic utility, that is $i^*=\argmax_{j\in V} u_{ij}$.

\begin{thm}\label{thm:approx}
For any $\epsilon\in[\frac{3}{4},1]$, the narcissistic-avaricious algorithm is linear time and produces an
$\epsilon$-Nash equilibrium with social welfare at least $(1-\epsilon)\cdot \Opt$.
\end{thm}
\begin{proof}
Consider the incentive guarantee for agent $i$ for the strategy profile $\z$ induced by the algorithm.
As $i$ votes with probability $z_{ii}=p$ and delegates to $i^*$ with probability $z_{ii^*}=(1-p)$,
his utility is 
\begin{align*}
		u_i(\z_i,\z_{-i}) &\ \geq\ z_{ii}u_{ii} + z_{i i^*} z_{i^* i^*}  u_{ii^*} \\
		&\ =\ p\cdot u_{ii} + (1-p) \cdot p\cdot u_{ii^*} \\
		&\ \geq\ (1-p)\cdot p\cdot u_{ii^*} \\
		&\ \geq\ (1-p)\cdot p \cdot u_i(\hat{\z_i},\z_{-i}) \qquad \forall \hat{\z_i}\in \mathbb{S}_i
\end{align*}
Hence $\z$ is an $\epsilon$-Nash equilibrium for 
$$\epsilon=(1-p(1-p))=1-p+p^2=\frac34 +(p-\frac12)^2$$
It follows that the narcissistic-avaricious algorithm can provide incentive guarantees only for $\epsilon\in[\frac{3}{4},1]$.
Further, by solving the corresponding quadratic equation, to obtain such an $\epsilon$-Nash equilibrium we simply 
select $p=\frac12\left(1+\sqrt{1-4(1-\epsilon)}\right)$.

Now, let's evaluate the social welfare guarantee for the narcissistic-avaricious algorithm.	
As above, we have
 \begin{align*}
	 	\SW(\z) &\ =\ \sum_{i\in V} u_i(\z_i,\z_{-i})\\
		&\ \geq\ \sum_{i\in V} (1-p)\cdot p\cdot u_{ii^*} \\
	 	&\ =\  (1-p)\cdot p\cdot\sum_{i\in V} u_{ii^*}\\
		&\ \geq\ (1-p)\cdot p\cdot \Opt 
 \end{align*}
 Since $p=\frac12\left(1+\sqrt{1-4(1-\epsilon)}\right)$ this gives
 $$\SW(\z) \ \ge\ \frac14\Big(1-(1-4(1-\epsilon)\Big)\cdot \Opt \ = \ (1-\epsilon)\cdot \Opt$$
Therefore, as claimed, the narcissistic-avaricious algorithm outputs a solution whose welfare is at least $(1-\epsilon)$
times the optimal welfare. 

Finally, observe that implementing the narcissistic-avaricious strategy requires that each agent $i$ simply 
computes $i^*=\argmax_{j\in V} u_{ij}$. This can be done for every agent in linear time in the size of the input.
\qed
\end{proof}

We emphasize two points concerning Theorem~\ref{thm:approx}. One, it only works for weaker incentive guarantees, namely 
$\epsilon\in[\frac{3}{4},1]$. Unlike the fixed point algorithm it does not work for the range $\epsilon\in (0,\frac34)$.
Two, the social welfare guarantee is $(1-\epsilon)$. This is a constant but, for the valid range $\epsilon\in[\frac{3}{4},1]$, it is much worse 
than the $\epsilon$ guarantee obtained by the fixed point algorithm. 
A very interesting open problem is to find a polynomial time algorithm that matches the optimal bicriteria
guarantees provided by Theorem~\ref{thm:bicriteria} and applies for all $\epsilon>0$.

\vspace*{0.5cm}

\noindent \textbf{Acknowledgements}: We thank Bundit Laekhanukit for interesting discussions. We are also grateful to the reviewers for suggestions which greatly improved the paper.

\bibliographystyle{splncs04}
\bibliography{LD}

\end{document}